\documentclass{amsart}

\newtheorem{thm}{Theorem}[section]

\newtheorem{lem}[thm]{Lemma}

\theoremstyle{definition}

\theoremstyle{remark}
\newtheorem{rem}[thm]{Remark}
\numberwithin{equation}{section}

\newcommand{\Real}{\mathbb R}

\newcommand{\To}{\longrightarrow}

\begin{document}

\title[The one-dimensional Schr\"{o}dinger-Newton equations]{The one-dimensional Schr\"{o}dinger-Newton equations}%
\author{Philippe Choquard* and Joachim Stubbe**}%
\address{*Institut for Theoretical Physics, EPFL, CH-1015 Lausanne, Switzerland}
\address{**EPFL, IMB-FSB, Station 8, CH-1015 Lausanne, Switzerland}%
\email{Philippe.Choquard@epfl.ch, Joachim.Stubbe@epfl.ch}%

\thanks{We are indebted to Marc Vuffray for providing us his numerical studies of the model.}

\date{22th june 2007}%
\begin{abstract}
We prove an existence and uniqueness result for ground states of
one-dimensional Schr\"{o}dinger-Newton equations.
\end{abstract}
\maketitle
\section{Introduction}
We consider the one-dimensional Schr\"{o}dinger-Newton system
\begin{equation}\label{SN}
    iu_t+u_{xx}-\gamma Vu=0,\quad V_{xx}=|u|^2
\end{equation}
which is equivalent to the nonlinear Schr\"{o}dinger equation
\begin{equation}\label{NLS}
    iu_t+u_{xx}-\frac{\gamma}{2}(|x|*|u|^2)u=0
\end{equation}
with nonlocal nonlinear potential
\begin{equation*}
   V(x)=\frac1{2} (|x|*|u|^2)(x,t)=\frac1{2}\int_{\Real}|x-y||u(t,y)|^2\;dy.
\end{equation*}
We are interested in the existence of nonlinear bound states of
the form
\begin{equation}\label{Bound states}
    u(t,x)=\phi_{\omega}(x)e^{-i\omega t}.
\end{equation}
The Schr\"{o}dinger-Newton system in three space dimensions
\begin{equation}\label{SN-3}
    iu_t+\Delta u-\gamma Vu=0,\quad \Delta V=|u|^2
\end{equation}
has a long standing history. With $\gamma$ designating appropriate
positive coupling constants it appeared first in 1954, then in
1976 and lastly in 1996 for describing the quantum mechanics of a
Polaron at rest by S. J. Pekar ~\cite{P54}, of an electron trapped
in its own hole by the first author ~\cite{L77} and of
selfgravitating matter by R. Penrose ~\cite{P96}. The
two-dimensional model is studied numerically in ~\cite{HMT2003}.
For the bound state problem there are rigorous results only for
the three dimensional model. In ~\cite{L77} the existence of a
unique ground state of the form \eqref{Bound states} is shown by
solving an appropriate minimization problem. This ground state
solution $\phi_{\omega}(x), \omega<0$ is a positive spherically
symmetric strictly decreasing function. In ~\cite{L80} the
existence of infinitely many distinct spherically symmetric
solutions is proven and a proof for the existence of anisotropic
bound states is claimed. So far, there are no results for the
one-dimensional model except for its semiclassical approximation
~\cite{CW2005}. One mathematical difficulty of the one-dimensional
problem is that the Coulomb potential does not define a positive
definite quadratic form (see below).

From numerical investigations of the problem we conjecture that in
the attractive case $\gamma>0$ equation \eqref{NLS} admits for
each $\omega>0$ infinitely many nonlinear bound states of the form
\eqref{Bound states} which means that subject to a normalization
condition $\int_{\Real}|u(t,x)|^2\;dx=N$ the model exhibits an
infinite discrete energy spectrum. In the present letter, we are
interested in the ground states of the model
\begin{equation}\label{Ground states}
    u(t,x)=\phi_{\omega}(x)e^{-i\omega t},\quad
    \phi_{\omega}(x)>0.
\end{equation}
We prove for any $\omega>0$ the existence of an unique spherically
symmetric ground state by solving an appropriate minimization
problem. We also prove the existence of an antisymmetric solution
by solving the same minimization problem restricted to the class
of antisymmetric functions.

\section{Mathematical Framework}

\subsection{Functional Setting} The natural function space $X$ for the quasi-stationary problem is given by
\begin{equation}\label{X}
    X=\{u:\mathbb{R}\to \mathbb{C}:\;\int_{\Real}|u_x|^2+|u|^2+|x||u|^2\;dx<\infty\}.
\end{equation}
Indeed, for each $u\in X$ the energy $E$ and the particle number
(or charge) $N$ associated to \eqref{NLS} given by
\begin{equation}\label{energy}
\begin{split}
    E(u)&=\int_{\Real}|u_x(x)|^2\;dx+\frac{\gamma}{4}\int_{\Real}\int_{\Real}|x-y||u(x)|^2|u(y)|^2\;dxdy\\
    &=T(u)+\frac{\gamma}{2}V(u)\\
    \end{split}
\end{equation}
and
\begin{equation}\label{charge}
    N(u)=\int_{\Real}|u(x)|^2\;dx,
\end{equation}
respectively, are well-defined quantities. In particular, the
energy functional $E:X\To\Real_0^{+}$ is of class $C^1$.

The space $X$ is a Hilbert space and by Rellich's criterion (see,
e.g. theorem XIII.65 of ~\cite{RS4}) the embedding
$X\hookrightarrow L^2$ is compact.
\subsection{Scaling properties} If $\phi_{\omega}(x)$ is a
solution of the stationary equation
\begin{equation}\label{sNLS-omega}
    -\phi_{\omega}''(x)+\frac{\gamma}{2}\bigg(\int_{\Real}|x-y||\phi_{\omega}(y)|^2\;dy\bigg)\;\phi_{\omega}(x)=\omega\phi_{\omega}(x),
\end{equation}
then $\phi_{1}(x)=\omega^{-1}\phi_{\omega}({x}/{\omega^{1/2}})$
solves
\begin{equation}\label{sNLS-1}
    -\phi_{1}''(x)+\frac{\gamma}{2}\bigg(\int_{\Real}|x-y||\phi_{1}(y)|^2\;dy\bigg)\;\phi_{1}(x)=\phi_{1}(x)
\end{equation}
and
\begin{equation}\label{scaling}
    E(\phi_{\omega})=\omega^{5/2}E(\phi_{1}),\quad N(\phi_{\omega})=\omega^{3/2}N(\phi_{1}).
\end{equation}
In addition, by the virial theorem
\begin{equation}\label{virial}
    4\omega N(\phi_{\omega})=20\,T(\phi_{\omega})=5\gamma V(\phi_{\omega}).
\end{equation}
\section{Ground states}
\subsection{Existence of ground states} We consider the following
minimization problem:
\begin{equation}\label{mini0}
    e_0(\lambda)=\inf \{E(u),u\in X,N(u)=\lambda\}.
\end{equation}
We note that the functional $u\to E(u)$ is not convex since the
quadratic form $f\to
\int_{\Real}\int_{\Real}|x-y|f(x)\bar{f}(y)\;dxdy$ is not positive
so that standard convex minimization does not apply. To see this
choose, for example, $f(x)=\xi_{[0,1]}(x)-\xi_{[1,2](x)}$ where
$\xi_{[a,b]}(x)$ denotes the characteristic function of the
interval $[a,b]$. For finite discrete systems it has been shown
that the associated matrix has only one positive eigenvalue
~\cite{K70}), which was computed in ~\cite{Ch75}).
\begin{thm} For any $\lambda>0$ there is a spherically symmetric
decreasing $u_{\lambda}\in X$ such that
$e_0(\lambda)=E(u_{\lambda})$ and $N(u_{\lambda})=\lambda$.
\end{thm}
\begin{proof} Let $(u_n)_n$ be a minimizing sequence for
$e(\lambda)$, that is $N(u_n)=\lambda$ and
$\underset{n\To\infty}{\lim} E(u_n)=e(\lambda)$. We also may
assume that $|E(u_n)|$ is uniformly bounded. Denoting $u^{*}$ the
spherically symmetric-decreasing rearrangement of $u$ we have (see
e.g. lemma 7.17 in ~\cite{LL01})
\begin{equation*}
    T(u^{*})\leq T(u), \quad N(u^{*})= N(u).
\end{equation*}
For the potential $V(u)$ we apply the following rearrangement
inequality:
\begin{lem}
Let $f,g$ be two nonnegative functions on $\Real$, vanishing at
infinity with spherically symmetric-decreasing rearrangement
$f*,g*$ , respectively. Let $v$ be a nonnegative spherically
symmetric increasing function. Then

\begin{equation}\label{Riesz2}
    \int_{\Real}\int_{\Real}f(x)v(x-y)g(y)\;dxdy\geq \int_{\Real}\int_{\Real}f^*(x)v(x-y)g^*(y)\;dxdy
\end{equation}
\end{lem}
\begin{proof}
If $v$ is bounded, $v\leq C$, then $(C-v)^*=C-v$ and by Riesz's
rearrangement inequality (lemma 3.6 in ~\cite{LL01}) we have
\begin{equation*}
    \int_{\Real}\int_{\Real}f(x)(C-v(x-y))g(y)\;dxdy\leq
    \int_{\Real}\int_{\Real}f^{*}(x)(C-v(x-y))g^{*}(y)\;dxdy.
\end{equation*}
Since
\begin{equation*}
    \int_{\Real}f(x)\;dx\int_{\Real}g(y)\;dy =
    \int_{\Real}f^{*}(x)\;dx\int_{\Real}g^{*}(y)\;dy
\end{equation*}
the claim follows. If $v$ is unbounded we define a truncation by
$v_n(x)=\sup{(v(x),n)}$ and apply the monotone convergence
theorem.
\end{proof}
By the preceding lemma we have
\begin{equation*}
    V(u^{*})\leq V(u)
\end{equation*}
since $|x|$ is an increasing spherically symmetric function.
Therefore we may suppose that $u_n=u^{*}_n$. We claim that
$u^{*}_n\in X$. Indeed, since $|x|$ is a convex function we have
\begin{equation*}
    V(u)\geq\frac{1}{2}\int_{\Real}\bigg|\int_{\Real}(x-y)|u(y)|^2\;dy \bigg||u(x)|^2\;dx
\end{equation*}
by Jensen's inequality and therefore
\begin{equation*}
    E(u^{*}_n)\geq
    T(u^{*}_n)+\lambda\frac{\gamma}{4}\int_{\Real}|x||u^{*}_n|^2\;dx
\end{equation*}
proving our claim. We may extract a subsequence which we denote
again by $(u^{*}_n)_n$ such that $u^{*}_n\to u^{*}$ weakly in $X$,
strongly in $L^2$ and a.e. where $u^{*}\in X$ is a nonnegative
spherically symmetric decreasing function. Note that $u^{*}\neq 0$
since $N(u^{*})=\lambda$. We want to show that $E(u^{*})\leq
\underset{n\To\infty}{\lim\inf}\;E(u^{*}_n)$. Since
\begin{equation*}
    T(u^{*})\leq
\underset{n\To\infty}{\lim\inf}\;T(u^{*}_n)
\end{equation*}
it remains to analyze the functional $V(u)$. First of all, we note
that for spherically symmetric densities $|u(x)|^2$ we have
\begin{equation*}
    V(u)=\int_{\Real}|x||u(x)|^2\bigg(\int_{-|x|}^{|x|}|u(y)|^2\;dy\bigg)\;dx.
\end{equation*}
Let
\begin{equation*}
    \eta(x)=\int_{-|x|}^{|x|}|u^{*}(y)|^2\;dy,\quad
    \eta_n(x)=\int_{-|x|}^{|x|}|u^{*}_{n}(y)|^2\;dy.
\end{equation*}
Then $\eta_n(x)\to\eta(x)$ uniformly since
\begin{equation*}
    ||\eta_n(x)-\eta(x)||_{\infty}\leq  ||u^{*}_n-u^{*}||_{2}\;
    ||u^{*}_n+u^{*}||_{2}\leq
    2\sqrt{\lambda}||u^{*}_n-u^{*}||_{2}.
\end{equation*}
Now
\begin{equation*}
\begin{split}
   &V(u^{*}_n)-V(u^{*})=\\
   &\int_{\Real}|x||u^{*}_n(x)|^2\big(\eta_n(x)-\eta(x)\big)\;dx+\int_{\Real}|x|\eta(x)\big(|u^{*}_n(x)|^2-|u^{*}(x)|^2\big)\;dx\\
   \end{split}
\end{equation*}
As $n\to\infty$ the first integral will tend to zero while the
second will remain nonnegative since the continuous functional
$\phi\to\int_{\Real}|x|\eta(x)|\phi(x)|^2\;dx$ is positive. Hence
\begin{equation*}
    V(u^{*})\leq
\underset{n\To\infty}{\lim\inf}\;V(u^{*}_n)
\end{equation*}
proving the theorem.
\end{proof}

\subsection{Uniqueness of ground states} As in ~\cite{L77} we need a
strict version of the rearrangement inequality for the potential
energy $V(u)$:
\begin{lem}
If $u\in X$ and $u^*(x)\notin \{e^{i\theta}u(x-a): \theta,
a\in\Real\}$, then we have the strict inequality:
\begin{equation}
    V(u)>V(u^{*})
\end{equation}
\end{lem}
\begin{proof}
We write
$|x|=-\frac1{1+|x|}+\frac{|x|^2+|x|+1}{1+|x|}=-g(x)+(|x|+g(x))$
where $g(x)$ is a spherically symmetric strictly decreasing
function and $g(x)+|x|$ is increasing. Then, from the strict
inequality for strictly decreasing functions (see ~\cite{L77}) we
have $V(u)>V(u^{*})$.
\end{proof}
After suitable rescaling the solution of the minimization problem
\eqref{mini0} satisfies the stationary equation \eqref{sNLS-1}
which is equivalent to the system of ordinary differential
equations
\begin{equation}\label{statSN-1}
    -\phi''+\gamma V\phi=\phi,\quad V''=\phi^2.
\end{equation}
Obviously, $\phi(x)>0$ for all $x$ and after another rescaling we
may assume that the pair $(\phi,V)$ satisfies the initial
conditions $\phi(0)>0, \phi'(0)=0, V(0)=V'(0)=0$. System
\eqref{statSN-1} is Hamiltonian with energy function given by
\begin{equation}\label{statSN-Hamiltonian}
    \mathcal{E}(\phi,\phi',V,V')=\phi'^2+\phi^2+\frac{\gamma}{2}V'^2-\gamma V\phi^2
\end{equation}
and $\mathcal{E}=\phi^2(0)$ for any symmetric solution.
\begin{thm}
The system \eqref{statSN-1} admits a unique symmetric solution
$(\phi,V)$ such that $\phi>0$ and $\phi\to 0$ as $|x|\to\infty$.
\end{thm}
\begin{proof} Suppose there are two distinct solutions $(u_1,V_1)$,
$(u_2,V_2)$ having the required properties. We may suppose
$u_2(0)>u_1(0)$. For $x\geq 0$ we consider the Wronskian
\begin{equation*}
    w(x)=u_2'(x)u_1(x)-u_1'(x)u_2(x).
\end{equation*}
Note that $w(0)=0$ and $w(x)\to 0$ as $x\to\infty$. It satisfies
the differential equation
\begin{equation*}
    w'=\gamma (V_2-V_1)u_1u_2.
\end{equation*}
Suppose $u_2(x)>u_1(x)$ for all $x\geq 0$. Then $V_2(x)>V_1(x)$
for all $x\geq 0$ since $(V_2-V_1)''=u_2^2-u_1^2>0$ and hence
$w'(x)>0$ for all $x> 0$ which is impossible. Hence there exists
$\bar{x}>0$ such that $\delta(x)=u_2(x)-u_1(x)>0$ for
$x\in[0,\bar{x}[$, $\delta(\bar{x})=0$ and $\delta'(\bar{x})<0$.
However, then $w(\bar{x})=\delta'(\bar{x})u_1(\bar{x})<0$, but
$w'(x)>0$ for all $x< \bar{x}$ which is again impossible.
\end{proof}
\subsection{Existence of antisymmetric ground states}
We consider the subspace $X^{as}$ of $X$ consisting of
antisymmetric functions, i.e. of functions $u$ such that
$u(-x)=-u(x)$. Repeating the arguments of the proof of theorem 2.1
we prove the existence of a solution of the minimization problem
\begin{equation}\label{mini0}
    e_1(\lambda)=\inf \{E(u),u\in X^{as},N(u)=\lambda\}
\end{equation}
which we conjecture to be the first excited state.
\begin{thm} For any $\lambda>0$ there is an antisymmetric $v_{\lambda}\in X$, positive for $x>0$ such that $e_1(\lambda)=E(v_{\lambda})$ and
$N(v_{\lambda})=\lambda$.
\end{thm}
\begin{proof} We may restrict the problem to the positive half-axis with Dirichlet boundary conditions. Then
\begin{equation*}
\begin{split}
    E(u)&=2\int_{0}^{\infty}|u_x(t,x)|^2\;dx+\frac{\gamma}{2}\int_{0}^{\infty}\int_{0}^{\infty}(|x-y|+|x+y|)|u(t,x)|^2|u(t,y)|^2\;dxdy\\
    &=2\int_{0}^{\infty}|u_x(t,x)|^2\;dx+\frac{\gamma}{2}\int_{0}^{\infty}\int_{0}^{\infty}|x-y||u(t,x)|^2|u(t,y)|^2\;dxdy\\
    &\qquad \qquad\qquad\qquad\qquad+\gamma\int_{0}^{\infty}|u(t,x)|^2\;dx\int_{0}^{\infty}|x||u(t,x)|^2\;dx.\\
    \end{split}
\end{equation*}
Let $(u_n)_n$ be a minimizing sequence for $e(\lambda)$, that is
$N(u_n)=\lambda$ and $\underset{n\To\infty}{\lim}
E(u_n)=e_1(\lambda)$.  We may suppose that the $u_n$ are
nonnegative on the positive half-axis. The rest of the proof
follows the same lines as the proof of theorem 3.1.
\end{proof}
\begin{rem}
As in theorem 3.3 we can show that the odd solution $\phi(x)$ of
\eqref{statSN-1} such that $\phi(x)>0$ for all $x>0$ which
corresponds to the initial conditions $\phi(0)=0, \phi'(0)>0,
V(0)=V'(0)=0$ is unique.
\end{rem}

\bibliographystyle{amsplain}

\begin{thebibliography}{}
\bibitem{P54} Pekar, S.I., Untersuchungen \"{u}ber die
Elektronentheorie der Kristalle, Akademie Verlag Berlin, 1954.
\bibitem{L77} Lieb, E.H., Existence and uniqueness of the minimizing solution of Choquard's nonlinear equation, Studies in Applied Mathematics 57, 93-105 (1977).
\bibitem{P96} Penrose, R., On gravity's role in quantum state reduction, Gen. Rel. Grav. 28, 581-600
(1996).
\bibitem{HMT2003} Harrison, R., Moroz, I. Tod, K.P., A numerical study of the Schr\"{o}dinger-Newton equation, Nonlinearity 16, 101-122 (2003).
\bibitem{L80} Lions, P.L., The Choquard equation and related questions, Nonlinear Analysis T.M.A. 4, 1063-1073 (1980).
\bibitem{L86} Lions, P.L., Solutions complexes d'\'{e}quations elliptiques semi-lin\'{e}aires dans $\mathbb{R}^n$, C.R. Acad. Sc. Paris t.302, S\'{e}rie 1 ,no. 19, 673-676 (1986).
\bibitem{CW2005} Choquard, Ph., Wagner, J., On a class of implicit solutions of the continuity and Euler's equatios for 1D systems with long range interactions, Physica D 40, 230-248 (2005).
\bibitem{RS4} Reed, M. and Simon, B., Methods of modern mathematical physics, Vol. IV, Analysis of operators, Academic Press 1978.
\bibitem{K70} Kac, M., Some probabilistic aspects of classical analysis, Am. Math. Mon. 77(6), 586-597 (1970).
\bibitem{Ch75} Choquard, Ph., On the statistical mechanics of One-dimensional Coulomb systems, Helv. Phys. Acta 48, 585-598 (1975).
\bibitem{LL01} Lieb, E.H. and Loss, M., Analysis, 2nd edition, Graduate Studies in Mathematics, vol. 14, AMS 2001.

\end{thebibliography}

\end{document}